\documentclass[conference]{IEEEtran}

\usepackage{cite}
\usepackage{amsmath,amssymb,amsfonts}
\usepackage{array}
\usepackage{multirow}
\usepackage{diagbox}
\usepackage{algorithm}
\usepackage[noend]{algorithmic}
\usepackage{graphicx}
\usepackage{textcomp}
\usepackage{xcolor}
\usepackage{amsthm}
\usepackage{braket}
\usepackage{qcircuit}
\usepackage{subfigure}
\usepackage[english]{babel}
\newtheorem{theorem}{Theorem}
\newtheorem{lemma}[theorem]{Lemma}
\usepackage{hyperref}
\usepackage{float}

\def\BibTeX{{\rm B\kern-.05em{\sc i\kern-.025em b}\kern-.08em
    T\kern-.1667em\lower.7ex\hbox{E}\kern-.125emX}}
\begin{document}

\title{
Assessing Quantum Layout Synthesis Tools via Known Optimal-SWAP Cost Benchmarks}

\author{
\normalsize Shuohao Ping$^1$, Wan-Hsuan Lin$^1$, Daniel Bochen Tan$^2$, Jason Cong$^1$\\
\normalsize $^1$Computer Science Department, University of California, Los Angeles, CA, USA,\\ 
$^2$Harvard University
 Cambridge, MA, USA \\
\normalsize \{sp1831, wanhsuanlin\}@ucla.edu; danieltan@g.harvard.edu; cong@cs.ucla.edu
\vspace{-3mm}
}

\date{}



\maketitle

\begin{abstract}
Quantum layout synthesis (QLS) is a critical step in quantum program compilation for superconducting quantum computers, involving the insertion of SWAP gates to satisfy hardware connectivity constraints.
While previous works have introduced SWAP-free benchmarks with known-optimal depths for evaluating QLS tools, these benchmarks overlook SWAP count—a key performance metric.
Real-world applications often require SWAP gates, making SWAP-free benchmarks insufficient for fully assessing QLS tool performance.
To address this limitation, we introduce QUBIKOS, a benchmark set with provable-optimal SWAP counts and non-trivial circuit structures. 
For the first time, we are able to quantify the optimality gaps of SWAP gate usages of the leading QLS algorithms, which are surprisingly large: LightSabre from IBM delivers the best performance with an optimality gap of 63x, followed by ML-QLS with an optimality gap of 117x.
Similarly, QMAP and t$\mathbf{\ket{\text{ket}}}$ exhibit significantly larger gaps of 250x and 330x, respectively. This highlights the need for further advancements in QLS methodologies.
Beyond evaluation, QUBIKOS offers valuable insights for guiding the development of future QLS tools, as demonstrated through an analysis of a suboptimal case in LightSABRE. 
This underscores QUBIKOS's utility as both an evaluation framework and a tool for advancing QLS research.

\end{abstract}

\section{Introduction}

Superconducting qubit has emerged as a promising solution for realizing large scale quantum computing.
However, the limited qubit connectivity in superconducting quantum computers restricts two-qubit gates to specific pairs of qubits.
In order to transform quantum circuits into hardware-executable forms, we need quantum layout synthesis (QLS) to insert SWAP gates that enable arbitrary two-qubit gates in the circuit.
However, these SWAP gates increase circuit size and depth, reducing overall fidelity. 
Thus, minimizing the number of inserted SWAP gates is a primary objective of QLS.
 

The QLS problem is proved to be NP-hard~\cite{siraichi_qubitallocation_2018}. 
The existing methods fall into two classes: exact approaches or heuristic algorithms. 
The exact approaches formulate a QLS problem into a constraint satisfaction problem and can obtain optimal results~\cite{wille2014optimal,bhattacharjee2019muqut,wille2019mapping,tan2020olsq, tan2021olsqga, zhang_time-optimal_2021, molavi2022satmap,nannicini2022optimal, lin2023olsq2} but it cannot scale up to solve problems with a hundred qubits.
To address scalability issues, several heuristic algorithms are proposed~\cite{siraichi_qubitallocation_2018,ho2018cirq,zulehner2018mapping_to_ibm_qx,web18-ibm-qiskit, zulehner_efficient_2019, siraichi_qubit_2019, li_sabre_2019,murali_formal_2019,sivarajah_tket_2020,kole_improved_2020,liu_notallswaparethesame_2022,wu_robust_2022,fan_QLSML_2022,huang2022reinforcement,park2022fsqm,huang2024ctqr,huang2024dear,lin2024mlqls}.
These methods are efficient, yet exhibit performance degradation as the problem size increases.
While there is a clear runtime advantage of heuristic algorithms over exact methods, the trade-off in solution quality remains uncertain due to the lack of benchmarks with known optimal SWAP counts.
Measuring the optimality gap using exact tools is an alternative, but this approach becomes impractical for large circuits due to scalability issues.
Therefore, having a benchmark set with known optimal SWAP counts is crucial for enabling solid comparison of QLS tools.

Few studies have attempted to establish such benchmark sets.
Tan and Cong introduced the QUEKO benchmark with known-optimal depth and zero SWAP gates\cite{tan2020queko}. 
However, this benchmark is limited in scope; it can be addressed by QLS tools employing subgraph isomorphism and fails to evaluate tool performance on circuits that require SWAP gates.
Li et al. proposed the QUEKNO benchmark~\cite{li2023QUEKNO}, which generates circuits requiring SWAP gates by deriving a series of mappings from an initial mapping and selected SWAP operations.
However, these circuits do not have known optimal SWAP counts, and the solutions are only claimed to be near-optimal without rigorously quantifying the optimality gap.
As a result, QUEKNO cannot measure the optimality gap of QLS tools.
To overcome the limitations of existing benchmarks, we propose QUBIKOS (\textbf{QU}antum \textbf{B}enchmark w\textbf{I}th \textbf{K}nown \textbf{O}ptimal \textbf{S}wap counts), which is the first benchmark set with a known optimal non-zero SWAP count. It permits a more comprehensive assesment of QLS tools and supports the development of future advancements in this area.


With QUBIKOS, we evaluate four QLS tools on four architectures and observe the optimality gaps. We find that as the size of the architecture grows, the optimality gap also grows from 1x to 233.97x even for the best performing tools. This shows an urgent need to have a more effective QLS tool to fully utilize the power of large systems. Furthermore, for architectures with a similar size, the structure of the architecture also affects the optimality gap. For instance, the optimality gap for IBM Rochester is seven times larger than Google Sycamore as a result of sparse connectivity and fewer axes of symmetry. 

QUBIKOS provides a qualitative evaluation on the performance of QLS tools by measuring the optimality gap between QLS results and optimal results. We can control the circuit size and number of optimal SWAP gates in the QUBIKOS circuit, which allows us to test the QLS tools under different conditions. QUBIKOS also offers insights on improving the performance of QLS tools by analyzing examples when QLS tools fails to find optimal SWAP gate counts. 


The rest of this paper is organized as follows. 
Section~\ref{sec:background} introduces the quantum layout synthesis problem.
Section~\ref{sec:construction} details the construction of QUBIKOS benchmark.
Section~\ref{sec:evaluation}
evaluates existing QLS tools with QUBIKOS, and Section~\ref{sec:conclusion} concludes the paper.

\section{Quantum Layout Synthesis}
\label{sec:background}
\begin{figure}[t]
\centering
\subfigure[]{
\includegraphics[width=0.5\linewidth]{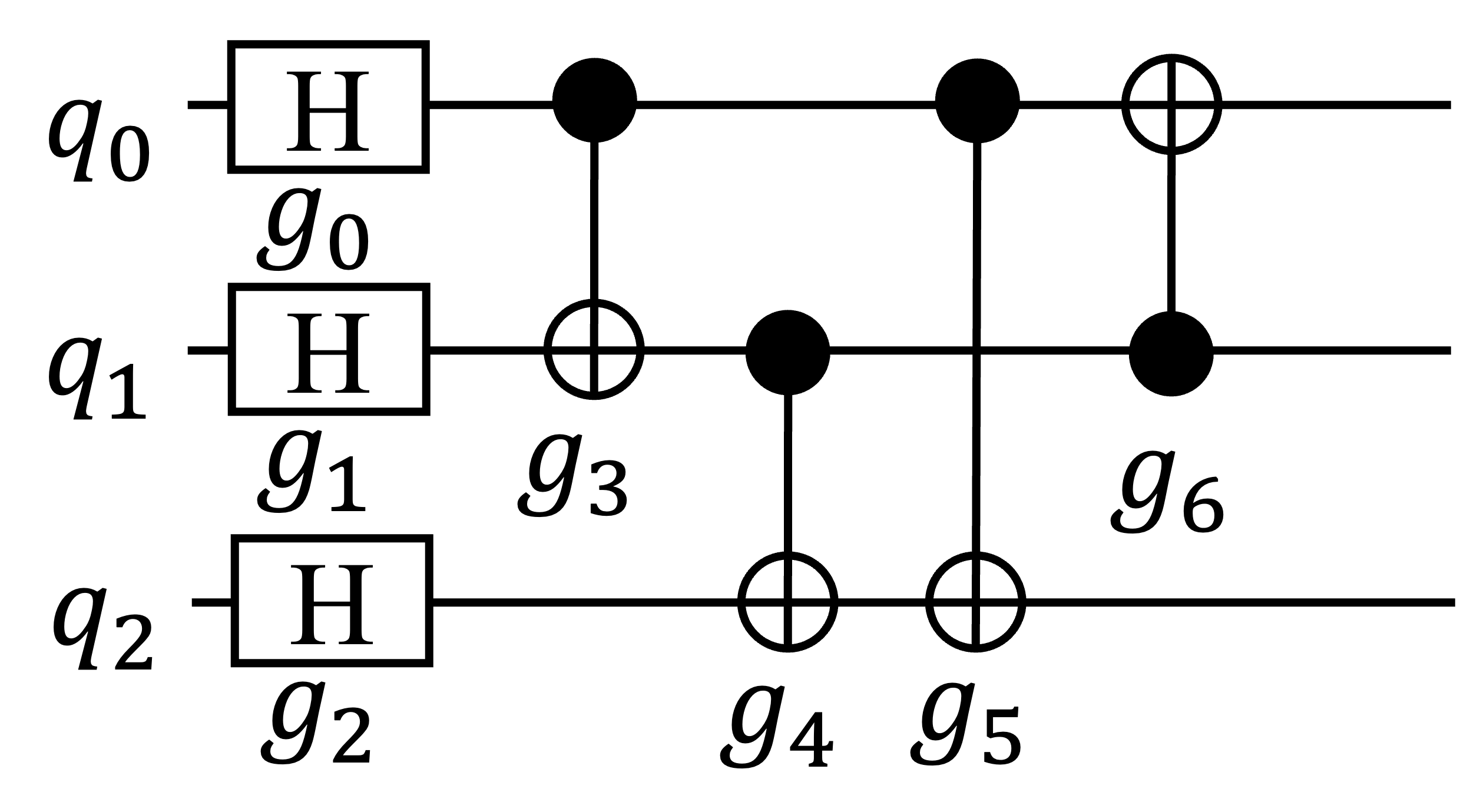}
}
\hfill
\raisebox{0.25cm}{
\subfigure[]{
\includegraphics[width=0.26\linewidth]{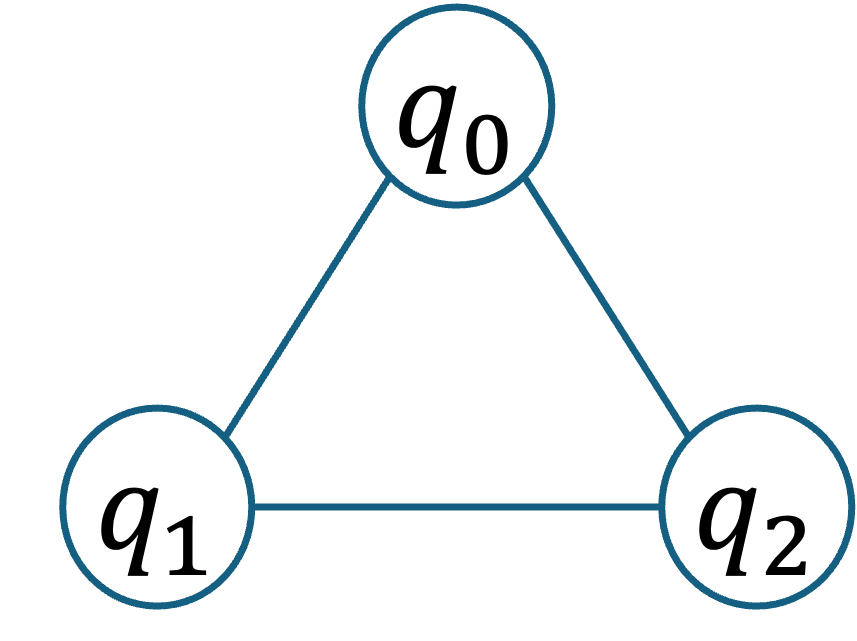}
}
}
\hfill
\subfigure[]{
\includegraphics[width=0.35\linewidth]{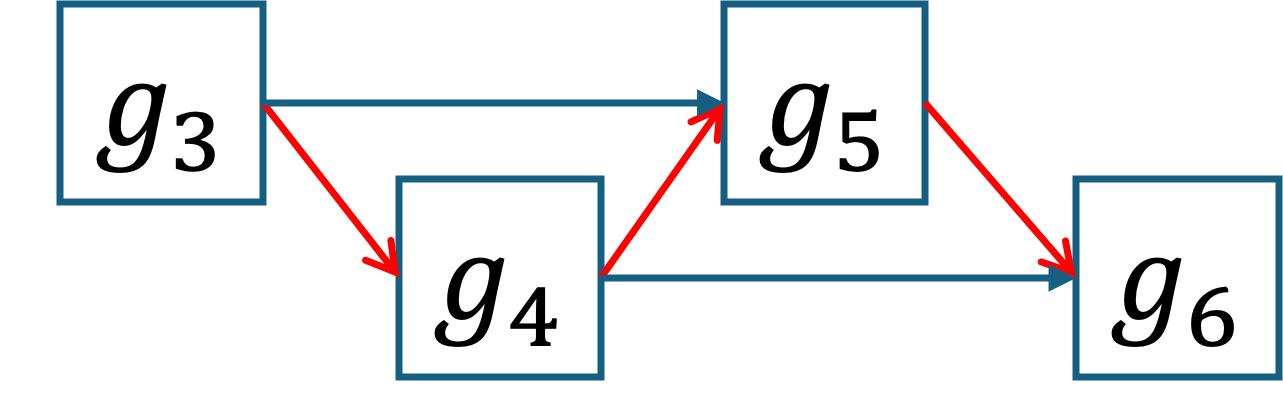}
}
\hfill
\subfigure[]{
\includegraphics[width=0.4\linewidth]{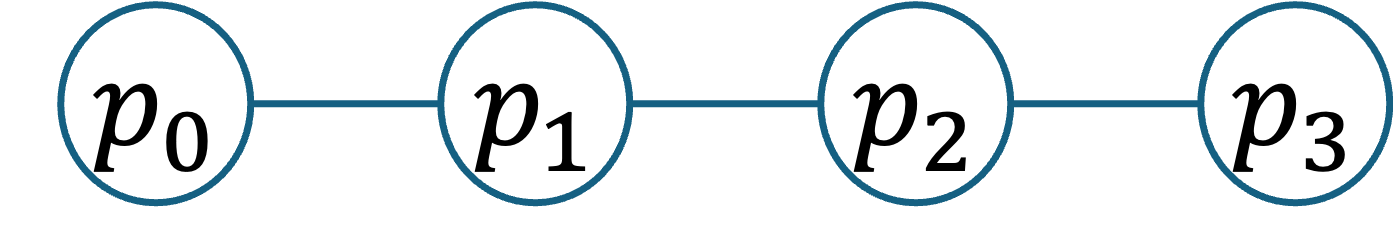}
}
\hfill
\subfigure[]{
\includegraphics[width=0.5\linewidth]{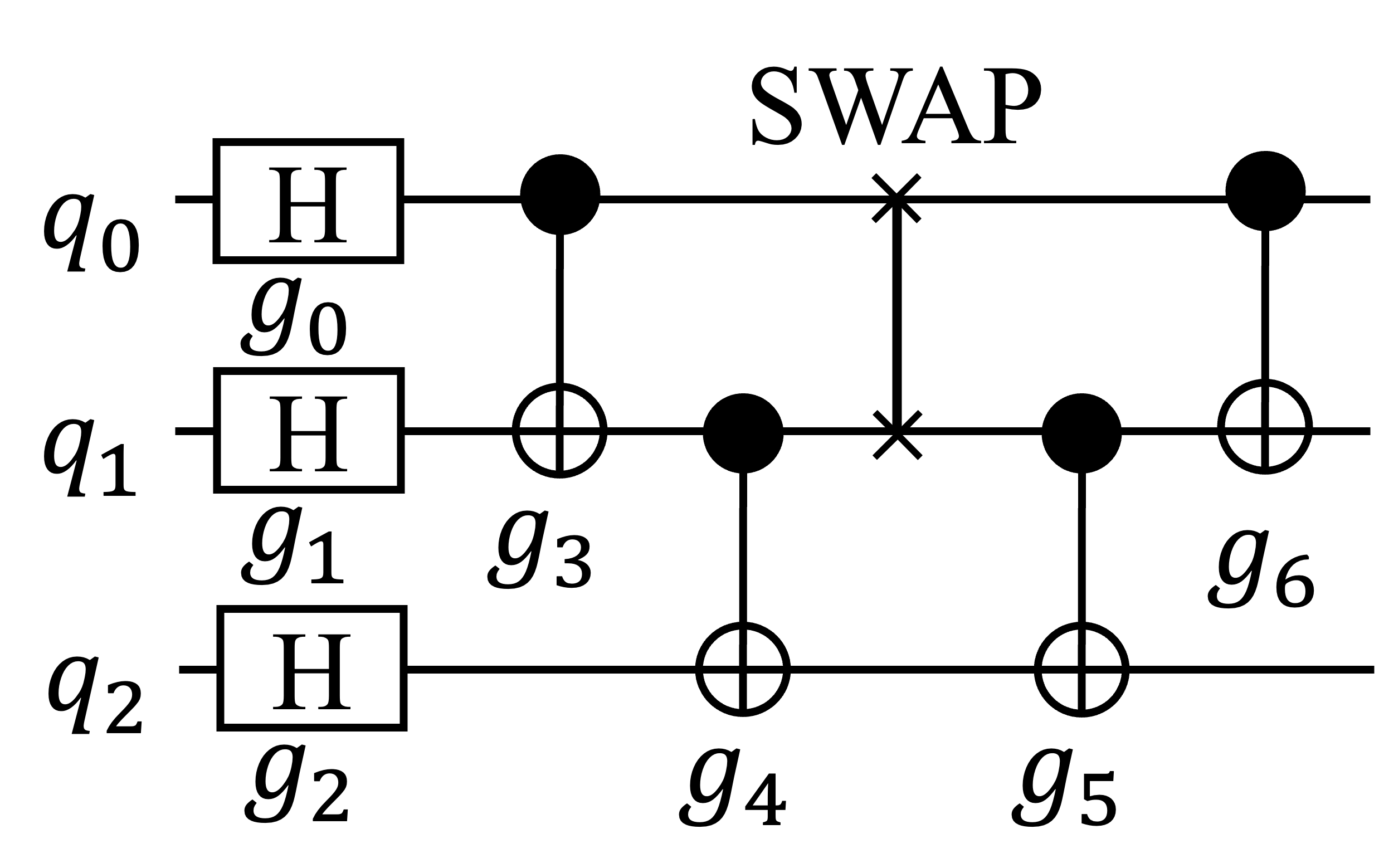}
}
\caption{(a) An example of a quantum circuit. (b) The interaction graph for the circuit shown in (a). (c) Gate dependency graph derived from the circuit. The red edges mark a dependence chain from $g_3$ to $g_5$.(d) A device coupling graph with 4 physical qubits in a line architecture. (b) A transpiled circuit of (a) with one inserted SWAP gate for the line architecture with mapping $q_0\rightarrow p_0, q_1\rightarrow p_1$, and $q_2\rightarrow p_2$. }
\label{fig:Simplecircuit}
\end{figure}

The QLS problem arises from the need to map gates in a quantum circuit to a superconducting quantum processor characterized by a coupling graph and perform necessary transformations to the circuit via inserting SWAP gates.
The inputs to the QLS problem are: \textit{quantum circuit} and \textit{coupling graph}. 

A \textit{quantum circuit} is defined as a sequence of
single and two-qubit gates.
In this paper, we denote
the set of program qubits by $Q$, a program qubit by $q\in Q$, the set of single-qubit gates by $\mathcal{G}_1$, the set of two-qubit gates by $\mathcal{G}_2$, and the overall gate set by
$\mathcal{G}=\mathcal{G}_1 \cup \mathcal{G}_2$. 
In addition, for a two-qubit gate $g(q,q')\in \mathcal{G}_2$, we use $g[0]=q$ and $g[1]=q'$ to denote the two program qubits $g$ operates on. Figure \ref{fig:Simplecircuit}(a) is an example of a quantum circuit where H gates are single qubit gates, and CNOT gates are two-qubit gates.

A \textit{coupling graph} $G_C(P,E_P)$ describes the connectivity of physical qubits for a given hardware where $P$ is a set of physical qubits and $E_P$ is a set of edges, where two-qubit gates can be performed. 
Figure~\ref{fig:Simplecircuit}(d) depicts a coupling graph. 
We define the function $\mathit{Neighbor(v,G)}$ to be the set of adjacent vertices of $v$ in $G$.
For example, the adjacent physical qubits of a physical qubit $p\in P$ is denoted as $\mathit{Neighbor(p,G_C)}$,


According to a quantum circuit $C$, we define the \textit{interaction graph} $G_I(Q,E_Q)$, where each program qubit is a node and an edge exists between $q$  and $q'$ if $g(q,q')\in \mathcal{G}_2$. 
Figure \ref{fig:Simplecircuit}(b) demonstrates the interaction graph of the circuit in Figure \ref{fig:Simplecircuit}(a). 
The neighbor of a program qubit $q\in Q$ in the interaction graph $G_I(Q,E_Q)$ is designated by $\mathit{Neighbor(q,G_I)}$. 

The \textit{gate dependency graph} $D(\mathcal{G}_2$,$E_G)$ is a directed acyclic graph that represents the execution order of two-qubit gates.
Single-qubit gates are excluded from $D$ since they do not impose connectivity constraints, and can be inserted back to the circuit after QLS.
Nodes in $D$ represent two-qubit gates in the circuit and an edge $(g,g')$ is drawn if $g'$ can be executed immediately after $g$. 
Figure \ref{fig:Simplecircuit}(c) shows an example of $D$ for the circuit in Figure \ref{fig:Simplecircuit}(a). 
There is an edge from $g_3$ to $g_4$ because they both act on $q_1$, and $g_4$ comes after $g_3$. 
In general, for any given gates $g,g'\in G_2$, if $g'$ must execute before $g$, there exists a path from $g'$ to $g$ in $D$. 
We define the set of previous gates of gate $g$ as $\mathit{Prev}(g)=\{g'|\exists\text{path from } g' \text{ to } g \text{ in }D\}$. 

The output of the QLS problem is a \textit{transpiled circuit} with inserted SWAP gates, and an initial mapping from program qubits to physical qubits $f:Q\rightarrow P$. 
Figure \ref{fig:Simplecircuit}(e) illustrates a valid QLS result using one SWAP to map the circuit in Figure \ref{fig:Simplecircuit}(a) to the coupling graph in Figure \ref{fig:Simplecircuit}(d). 
A transpiled circuit can be written in the form of $C_0\cdot T_0 \cdot C_1 \cdot T_1 \cdot... \cdot T_{n-1} \cdot C_n$, where
$C_i$ is a partial gate sequence with
$C_0\cdot C_1 \cdot... \cdot C_n=C$, 
and $T_i$ is a SWAP gate that implements mapping transformation. 


\section{QUBIKOS Circuit Construction}
 \label{sec:construction}
To construct a circuit with an optimal SWAP count of $n$ given a coupling graph $G_C$, 
the first step is to generate $n$ circuit segments.
Each requires exactly one SWAP gate under all possible mappings, which means the interaction graph is not isomorphic to any subgraph of $G_C$,
as such isomorphisms would allow execution without SWAP insertion.
Figure~\ref{fig:example2}(c) is a non-isomorphic interaction graph to the coupling graph $G_C$ displayed in Figure~\ref{fig:example2}(a) as the interaction graph contains a node with degree 5, while the maximum degree in $G_C$ is 4.

However, simply concatenating these graphs does not guarantee an optimal SWAP count. 
Thus, the second step is to construct gate dependencies both within and between the interaction graphs and enforce that gates from different interaction graphs execute in serial.
This guarantees that when we concatenate the segments to construct the final circuit, the corresponding optimal SWAP count equals the sum of the optimal SWAP counts for each segment. 

Figure~\ref{fig:exampleDAG} shows a gate dependency graph of a circuit that requires 2 SWAP gates. 
Here, $g_4$ and $g_{15}$ partition the circuit into two sets: $g_0-g_4$ and $g_5-g_{15}$, with no gates from different sets executed in parallel. 
These sets form two non-isomorphic interaction graphs that must execute serially, ensuring each requires one SWAP gate to resolve connectivity,
and no single SWAP can resolve conflicts in both graphs.

Repeating this process $n$ times, we can generate a circuit $C$ with $n$ SWAP gates. 
Finally, we can insert additional gates that do not increase the optimal SWAP count, allowing for arbitrary circuit structures and making it more challenging to solve. 
The following sections provide detailed examples of the subroutines.

\begin{figure}
\centering
\subfigure[]{\includegraphics[width=0.7\linewidth]{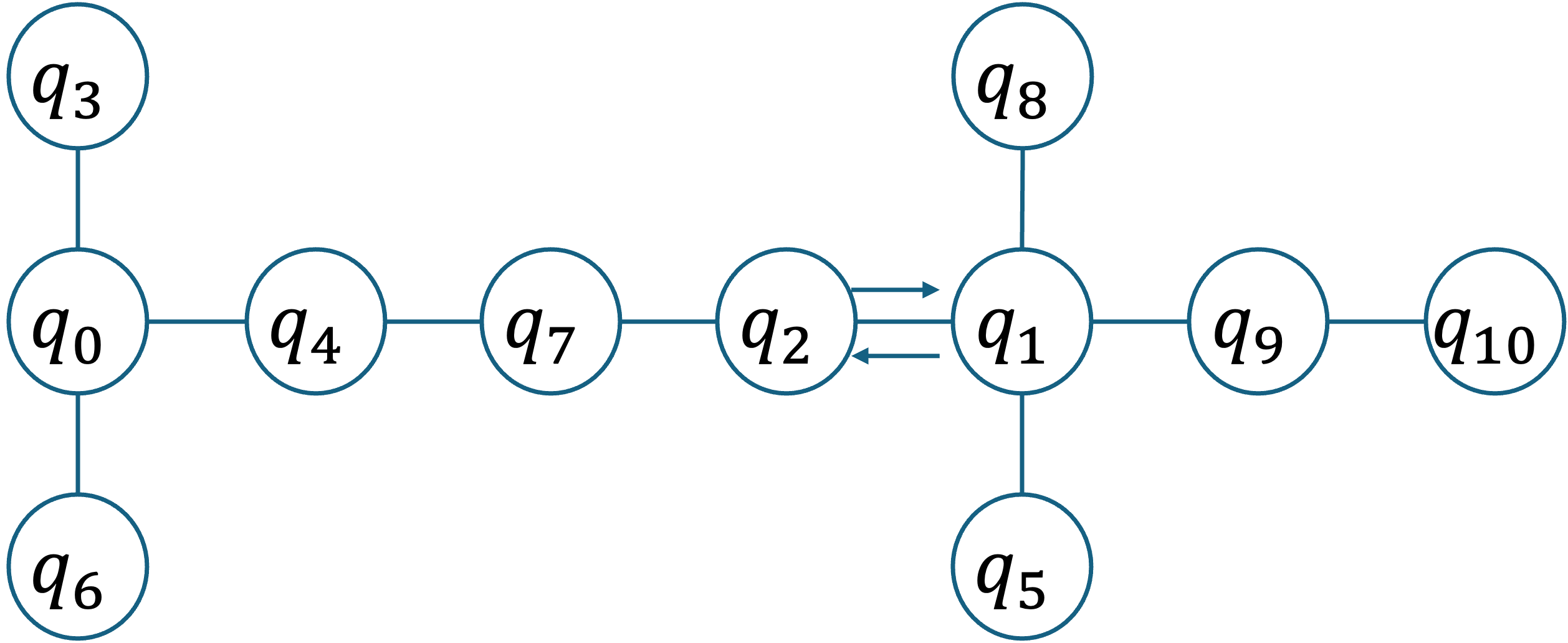}
}
\subfigure[]{\includegraphics[width=0.7\linewidth]{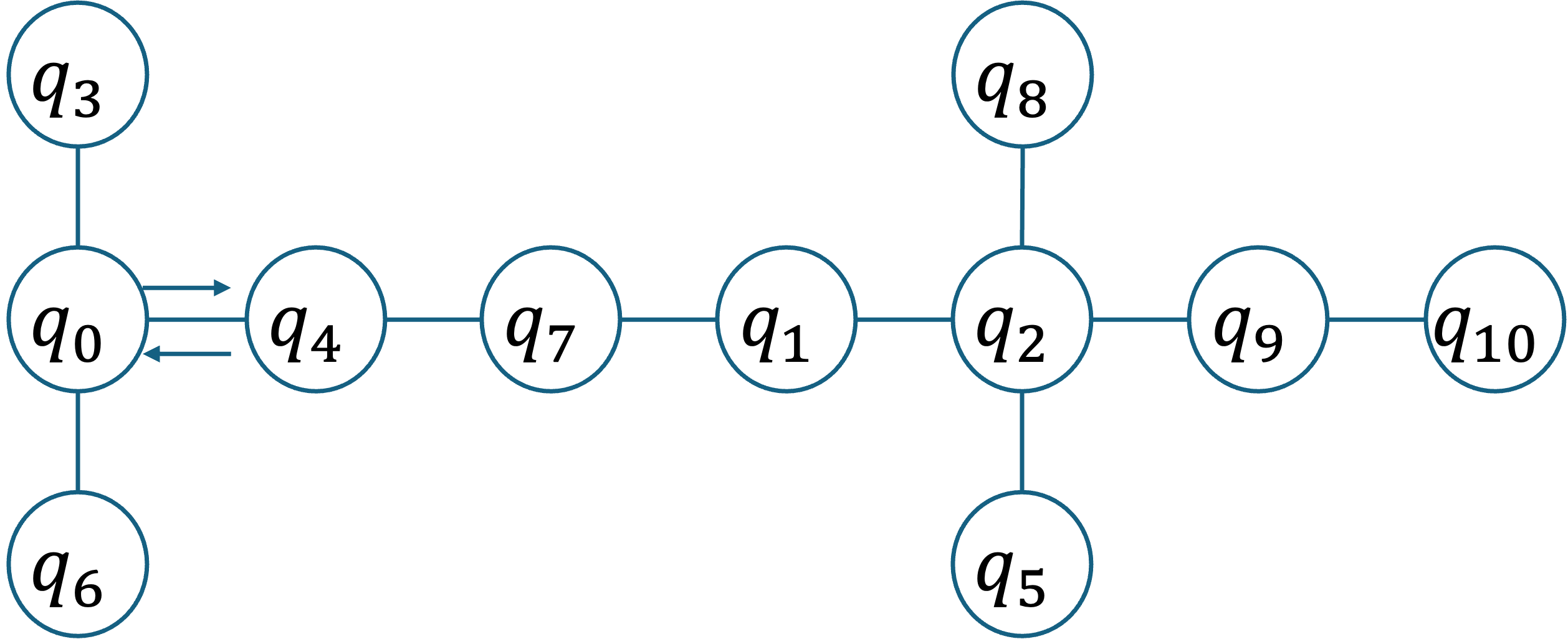}
}

\subfigure[]{
\includegraphics[width=0.26\linewidth]{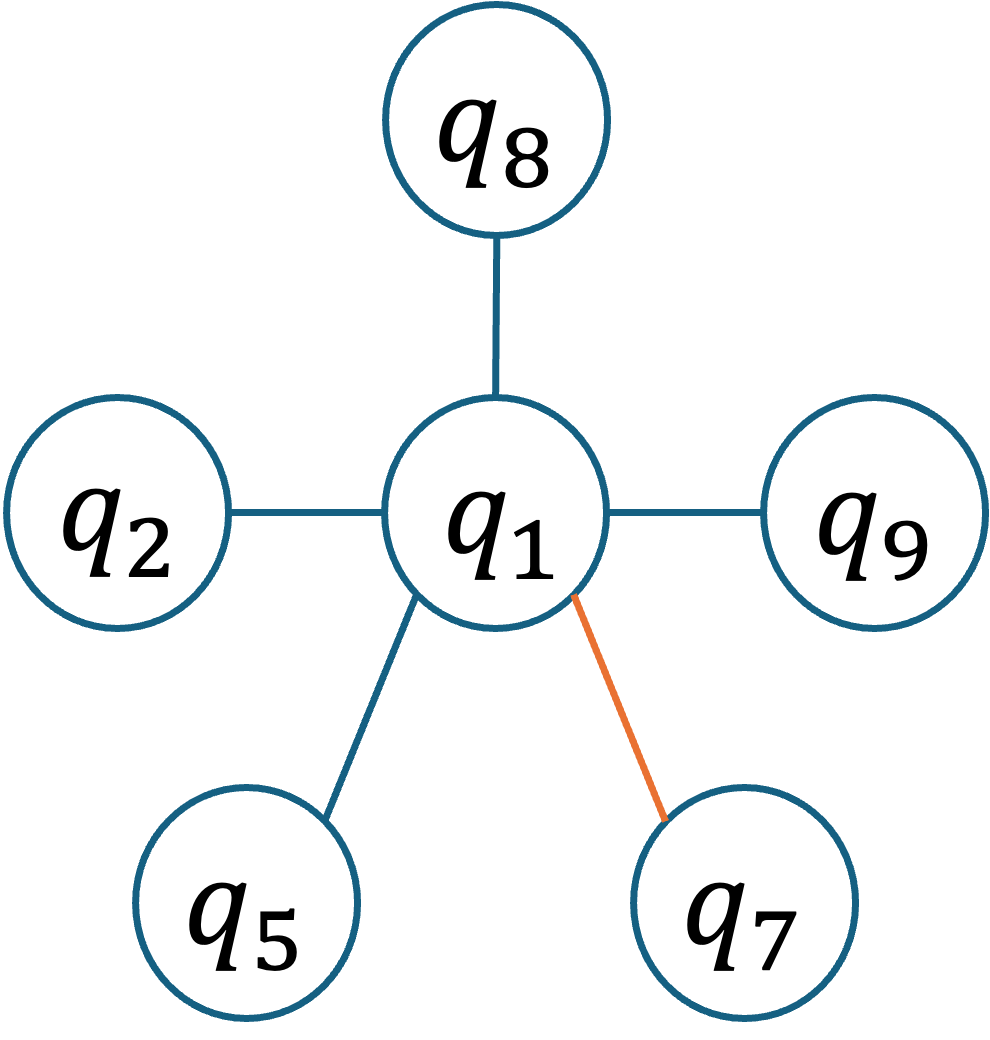}}
\subfigure[]{\includegraphics[width=0.61\linewidth]{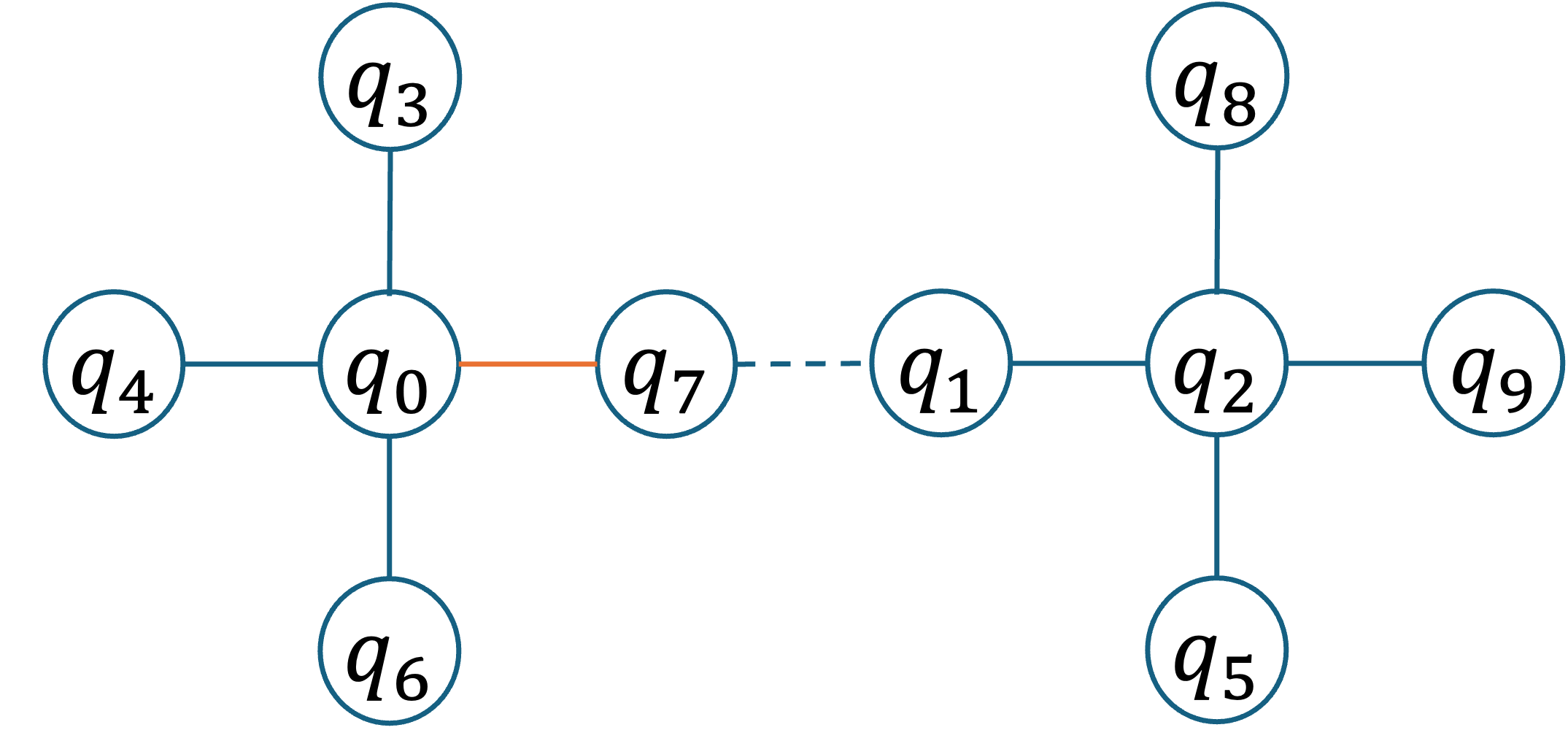}
}

\caption{(a) The device coupling graph $G_C$ with a random initial mapping. Arrows represent the first SWAP operation. (b) The device coupling graph with qubit mapping after SWAP($q_1,q_2$) on the initial mapping. Arrows indicate the second SWAP operation. (c) The interaction graph $G_{I_1}(Q,E_{Q_1})$ that enforces SWAP($q_0,q_4$). The orange edge is the special gate.
(d) The interaction graph $G'_{I_2}(Q,E'_{Q_2})$ that only considers gates after the first SWAP. The orange edge is the special gate in $E_{Q_2}$. The dotted line is the inserted gate. Without it, the coupling graph is $G_{I_2}$.}
\label{fig:example2}
\end{figure}

\subsection{Non-isomorphic Interaction Graph Generation}
\label{subsec:isoigg}
The first step constructs non-isomorphic interaction graphs relative to a given architecture $G_C$. 
Starting with a random initial mapping, we select an adjacent qubit pairs for SWAP gate. 
The SWAP is designed to enable at least one qubit in the pair to interact with a new neighbor, ensuring the SWAP is essential. 
If no new interactions are introduced, the SWAP is redundant.
Such edges always exist unless the coupling graph is fully connected. 
Figure~\ref{fig:example2}(a) illustrates an example of the initial mapping on the coupling graph with SWAP gates on $(q_1, q_2)$, enabling new interactions. For selecting subsequent SWAP gates, we need to transform the mapping based on the selected SWAP and choose an adjacent qubit pair in the new mapping that satisfies the condition discussed above. Figure~\ref{fig:example2} shows the new mapping after SWAP$(q_1,q_2)$, and we choose edge $(q_0,q_4)$ for the next SWAP.

To enforce a SWAP between $q_1$ and $q_2$, we construct an interaction graph based on $q_1$.
First, $q_1$ must interact with all its current neighbors, ensuring it has no free edges to form new interactions.
This requires a two-qubit gate on $g(q_1,q_2)$, $g(q_1,q_5)$, $g(q_1,q_8)$, and $g(q_1,q_9)$.
Next, we identify a qubit that is not a neighbor of $q_1$ before the SWAP gate but becomes one afterward, e.g., $q_7$, and add the \emph{special gate} $g(q_1, q_7)$ in the circuit.
The resulting interaction graph is displayed in Figure~\ref{fig:example2}(c), and is not isomorphic to any subgraph of $G_C$.
Even though the circuit before the special gate does not need a SWAP gate, it forces the circuit to require one prior to the its execution. 
Any gate in the interaction graph could serve as the special gate, provided the remaining four gates form an isomorphic subgraph. 
The special gate must execute last to ensure all other gates in the interaction graph are completed beforehand.
Otherwise, it may alter the interaction graph by removing an edge and eliminate the need for a SWAP gate. 
Section~\ref{subsec:build_dependency} will present how to impose the dependency constraint.

Alternatively, we could use $q_2$ to build the interaction graph.
However, simply having $q_2$ interact with all its neighbors plus one new neighbor qubit is insufficient.
In this case, $q_2$ would have a degree of 3 in the interaction graph, and a mapping could exist that avoids SWAP gates by placing $q_2$ on a degree 3 or degree 4 node in $G_C$.
To prevent this, all program qubits mapped to physical qubits with degrees higher than 2  must also interact with all their neighbors. 
This ensures all higher-degree physical nodes are occupied, making it impossible to find a new mapping that avoids SWAP gates by placing $q_2$ on a higher-degree physical qubit.

Next, we want to construct a circuit $C$ to compel SWAP$(q_0,q_4)$. 
After SWAP$(q_1,q_2)$, the mapping is updated as in Figure~\ref{fig:example2}(b).
Using $q_0$ as the starting point for constructing the non-isomorphic graph, we ensure $q_0$ interacts with all its neighbors, generating gates $g(q_0,q_3)$ $g(q_0,q_6)$, and $g(q_0,q_4)$.
Then, $g(q_0,q_7)$ serves as the special gate that enforce SWAP$(q_0,q_4)$.
In addition, $q_2$ must interact with all its neighbors since $q_2$ is mapped to a physical qubit whose degree is larger than the one that $q_0$ is mapped to. Figure~\ref{fig:example2}(d) shows the interaction graph of $G_{I_2}$. Note that the dotted line is current not part of the interaction graph, but will be added in the next subroutine.



\begin{figure}
\includegraphics[width=0.5\textwidth]{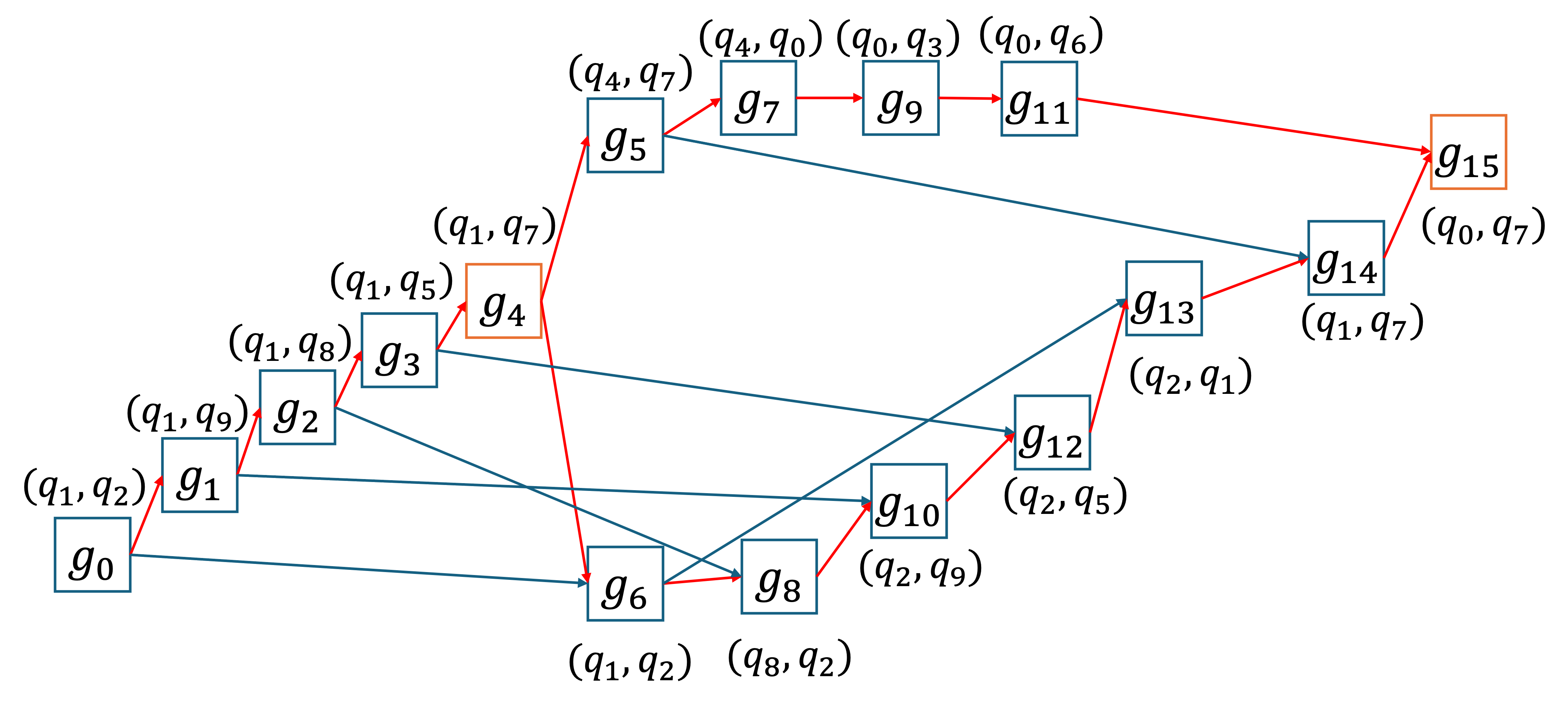}
\caption{Gate dependency graph $D$ for circuit $C$. The orange boxes represent the special gates. The red edges show the required dependence relation for $g_4$ and $g_{15}$. 
}
\label{fig:exampleDAG}
\end{figure}

\subsection{Gate Dependency Relation}
\label{subsec:build_dependency}
Given the interaction graphs $G_{I_1}, G_{I_2}$ generated by the first step, we build the QUBIKOS circuits by enforcing the gate dependency within and between interaction graphs. 
Considering the gates in $G_{I_1}$, 
since all these gates act on $q_0$, we can reorder them so that the special gate depends on all other gates, ensuring that the special gate forces a SWAP operation.
For example, in the gate sequence $C_1=[g_0(q_1,q_2),g_1(q_1,q_9),g_2(q_1,q_8),g_3(q_1,q_5),g_4(q_1,q_7)]$, the special gate $g_4(q_1,q_7)$ has to execute after all other gates. 

Reordering alone does not always ensure the desired gate dependencies. 
For instance, in $G_{I_2}$, the special gate $(q_0, q_7)$ cannot be made to execute after $(q_1, q_2)$ simply by reordering gates in $E_{Q_2}$
because the two disconnected components in $G_{I_2}$ can always execute in parallel regardless of gate order. 
To address this, we must add edges to $G_{I_2}$ to connect the two components.
For example, we can add edge $(q_1, q_7)$ as shown in Figure~\ref{fig:example2}(d), and denote the new interaction graph as $G'_{I_2}$.

A correct gate order necessitates that the first special gate $(q_1, q_7)$ is executed before all gates in $E'_{Q_2}$, and the second special gate $(q_0, q_7)$ execute after the other gates in $E'_{Q_2}$.
To satisfy the first dependency constraint, we perform a breadth-first search (BFS) on $G'_{I_2}$, starting from $q_1$ and $q_7$ while ignoring the edge $(q_0, q_7)$. 
The order of visited edges from the BFS will establish the desired gate dependencies. 
In this case, the BFS result is $[ (q_7,q_1),\allowbreak (q_0,q_4), (q_1,q_2),\allowbreak (q_0,q_3),\allowbreak (q_2,q_8),\allowbreak (q_0,q_6),\allowbreak (q_2,q_9),\allowbreak (q_2,q_5)].$
By prepending $(q_1,q_7)$ to the BFS sequence, all gates will depend on it. 
To satisfy the second constraint, we run BFS on the same graph, starting from $q_0$ and $q_7$, reverse the order of visited edges, and append $(q_0, q_7)$ at the end.
Concatenating these two sequences yields a circuit that meets both dependency constraints. 
While this method does not produce the smallest possible circuit, it is valid. 
For simplicity, we use a smaller sequence as an example.
$C_2 = [g_5(q_4, q_7),\allowbreak g_6(q_1, q_2),\allowbreak g_7(q_4, q_0), g_8(q_8, q_2), g_9(q_0, q_3), \allowbreak
      g_{10}(q_2, q_9),\allowbreak g_{11}(q_0, q_6), \allowbreak g_{12}(q_2, q_5),\allowbreak g_{13}(q_2, q_1),\allowbreak g_{14}(q_1, q_7),\allowbreak g_{15}(q_0, q_7)]$.

Now, we can build a circuit $C$ with an optimal SWAP count two by concatenating $C_1$ and $C_2$, i.e., $C=C_1\cdot C_2$.
The gate dependency graph is illustrated in Figure~\ref{fig:exampleDAG}, showing that gates from two interaction graphs must be executed in serial.
We call $C$ the \emph{backbone} of the benchmark circuit and all gates between two special gates, including the later special gate, a \emph{section} of the backbone circuit that enforce one SWAP gate. 

We can introduce additional gates into the backbone without affecting the optimal SWAP gate count. 
For example, we can add $(q_9,q_{10})$ anywhere in $C$ since it can be executed under any mapping. 
The key constraint is ensuring that additional gates do not introduce extra SWAP gates.
For instance, $(q_2,q_7)$ can only be inserted before $g_4$ since $q_2$ and $q_7$ are not neighbors after the first SWAP gate. 

\subsection{Challenges in Finding Optimal Solution for QUBIKOS}
Unlike QUEKO\cite{tan2020queko}, the optimal transpiled circuits for QUBIKOS benchmarks cannot be obtained using subgraph isomorphism algorithms, e.g., VF2\cite{VF2}.
One may identify the special gates within the benchmark circuit by tracking the degrees of the interaction graph, and the special gates will partition the circuit into multiple sections that have an isomorphic subgraph of $G$.
Nevertheless, running subgraph isomorphism algorithms on individual sections may not yield an optimal solution.
For instance, applying VF2 to the first four gates only provides mapping information for a subset of qubits, and may not be optimal for the subsequent sections.
As the optimal initial mapping for all qubits can only be determined by examining the global structure of the circuit, this benchmark is effective for testing modern QLS tools.
\begin{algorithm}
\caption{Non-isomorphic interaction graph generation}
\label{alg:1}
\begin{algorithmic}[1]
\REQUIRE Coupling graph $G_C$, mapping function $f:Q\rightarrow P$, and an edge $(p,p')$
 a SWAP operation $\mathit{SWAP}(p_1,p_2), p_1,p_2\in P$
\ENSURE A list of edges $S$ and a special gate $g$ 
 forming the interaction graph.
\STATE $p'' \leftarrow $ randomly select a qubit that is a neighbor of $p$ but not a neighbor of $p'$, excluding $p'$
 \STATE $p' \leftarrow $ randomly select a qubit that is not a common neighbor of both $p_1,p_2$ excluding themselves. 
 \STATE $p\leftarrow$ select from $\{p_1,p_2\}$ such that $p' \notin Neighbor(p,G_C)$
\STATE $S \leftarrow [\ ]$
\STATE $g \leftarrow (f^{-1}(p),f^{-1}(p''))$
\FOR{$(p_a,p_b) \in E$} 
\IF{$q\in(p_a, p_b)$ or $\max(\deg(p_a), \deg(p_b))>\deg(p)$}
\STATE S.append($(f^{-1}(p_a),f^{-1}(p_b))$)
 \ELSIF{$\deg(p_a)>\deg(p)$ or $\deg(p_b)>\deg(p)$}
 \STATE S.append($(f^{-1}(p_a),f^{-1}(p_b))$)
\ENDIF
\ENDFOR
\RETURN $S,g$
\end{algorithmic}
\end{algorithm}


\begin{algorithm}
\caption{Backbone circuit generation}
\label{alg:2}
\begin{algorithmic}[1]
\REQUIRE Coupling graph $G_C(P,E_P)$, a mapping function $f$, a set of edge $S$, a prior special gate $g_1$, and a subsequent special gate $g_2$. 
\ENSURE A section of backbone circuit $C_{sub}$.
\STATE $C_{sub}\leftarrow []$
\IF{$g_1$ != Null}
\STATE $G_{sub}\leftarrow \text{findConnectedSubgraph}(G_C,f,[g_1]+S)$
\STATE $q\leftarrow$ randomly select from $\{g_1[0],g_1[1]\}$
\STATE $C_{prior}\leftarrow \mathit{BFS}(G_{sub},q)$
\STATE $C_{sub}\leftarrow C_{sub}+C_{prior}+S$
\ENDIF
\STATE $G'_{subI}\leftarrow \text{findConnectedSubgraph}(G_C,f,S+[g_2])$
\STATE $q'\leftarrow$ randomly select from $\{g_2[0],g_2[1]\}$ 
\STATE $C_{later}\leftarrow BFS(G'_{sub},q').reverse()$
\STATE $C_{sub}\leftarrow C_{sub}+C_{later}+[g_2]$
\RETURN $C_{sub}$
\end{algorithmic}
\end{algorithm}


 \begin{algorithm}
 \caption{Circuit generation}
 \label{alg:3}
 \begin{algorithmic}[1]
 \REQUIRE Coupling graph $G(V,E)$, number of SWAP gate $n$, total number of two-qubit gates $N$.
 \ENSURE A circuit $C$ and a transpiled circuit $C_{ans}$.

 \STATE $f_\mathit{init} \leftarrow $ randomly generate $f:V'\rightarrow V$
 \STATE $f_{cur}\leftarrow f_\mathit{init}$
 \STATE Generate a sequence of edges $S$ with length $n$
 \STATE $g_{prev},C,C_{ans}\leftarrow Null,\{\},\{\}$
 \FOR{$(q,q') \in S$}
 \STATE $S,g_{cur}\leftarrow \text{NonIsomorphic}(G,f_{cur},(q,q'))$ 
 \STATE $C_{sub}\leftarrow \text{BuildSection}(G,S,f_{cur},g_{prev},g_{cur})$
 
 \STATE $C\leftarrow C+C_{sub}$
 \STATE $C_{ans}\leftarrow C_{ans}+C_{sub}[:-1]+\{\mathit{SWAP}(q,q'),C_{sub}[-1]\}$
 \STATE $g_{prev}\leftarrow g_{cur}$
 \STATE Update $f_{cur}$ by $s$
 \ENDFOR
 \WHILE{$C.size < N$}
 \STATE Insert a gate into $C$ and $C_\mathit{ans}$
 \ENDWHILE
\RETURN $C,C_{ans}$
 \end{algorithmic}
 \end{algorithm}



\subsection{Proof of Optimality}
In this section, we prove the backbone circuit $C$ constructed by Algorithm~\ref{alg:3} requires at least $n$ SWAP gates when executed on the coupling graph $G_{C}$, and there exists a solution with $n$ SWAP gates.
Then, we show inserting redundant gates does not affect the optimal SWAP count.

Since Algorithm \ref{alg:1} will be called $n$ times, let $L_{special}=[g'_0,g'_1,...,g'_{n-1}]$ be a list of special gates generated from the first to the last call. 
 

\begin{lemma}
\label{th:1}
 $S\cup \{g\}$ generated by algorithm \ref{alg:1} forms an interaction graph that is not isomorphic to any subgraph of $G_C(P,E_P)$.
\end{lemma}
\begin{proof}
Let $G_I(Q,E_Q)$ be the interaction graph of $S\cup \{g
\}$. For $q\in Q$, we define $S_1=\{p |\deg(p)\geq \deg(q),p \in P\}$ and $S_2=\{q' |\deg(q')\geq \deg(q), q' \in Q \}$. 
We know $\exists q \in \{g[0],g[1]\}$ such that $|S_1| > |S_2|$. 
This means we need to map at least one program qubit in $Q$ to a physical qubit $P$ with less degree no matter what mapping we choose. 
Thus, $G_I$ is not isomorphic to any subgraph of $G_C$.
\end{proof}

\begin{lemma}
\label{th:2}
 All gates in circuit $C_{sub}$ generated by Algorithm \ref{alg:2} has to execute after the prior special gate $g_1$ but before the subsequent special gate $g_2$.
\end{lemma}
\begin{proof}
Let $G_I(Q,E_Q)$ be the interaction graph of $C_\mathit{sub}$. 
By construction,  $\forall g \in C_\mathit{sub}$, $\exists q\in \{g[0],g[1]\},q'\in \{g_1[0],g_1[1]\}$ such that BFS finds a path from $q'$ to $q$. Since the order of the gates is the same as the edges visited by BFS, we will have a path from $g_1$ to $g$ in DAG by finding the corresponding gate label of the edges visited by BFS. Hence, $g_1$ must execute before $g$.

For the second property, suppose $g \in C_{sub}$ and $g\neq g_2$. By construction, we have $\exists q\in \{g[0],g[1]\},q'\in \{g_1[0],g_1[1]\}$ such that BFS finds a path from $q'$ to $q$. But the order of the gate is reversed after BFS, so we can find a path from $g$ to $g_2$ in DAG. Hence, $g_2$ must execute after $g$.
\end{proof}


\begin{lemma}
\label{th:3}
$\forall 0\leq i \leq n-1, \forall g\in C_i,\forall g'\in C_{i+1}: g\in Prev(g')$
\end{lemma}
\begin{proof}
First, we know $\exists g_{s_i}\in S_{special}:g_{s_i}\in C_{i+1}$. Suppose $g\in C_i$ and $g'\in C_{i+1}$. If $g'$ is $g_{s_i}$, then by Lemma \ref{th:2}, we have $g\in Prev(g')$. Otherwise, by Lemma \ref{th:2}, we know there exists a path from $g$ to $g_{s_i}$ and a path from $g_{s_i}$ to $g'$ in DAG. Hence, there exists a path from $g$ to $g'$. Hence, $g\in Prev(g')$. 
\end{proof}
\begin{theorem}
\label{th:4}
 The backbone of circuit $C$ needs at least $n $ SWAP gates to execute.
\end{theorem}
\begin{proof}
Assume the contrary that there exists a transpiled circuit of $C$ with $k<n$ SWAP gates. We denote such transpiled circuit as $C'=C'_0\cup T'_0 \cup ... \cup T'_{k-1} \cup C'_{k}$ with $\forall 0\leq i \leq k-1:|T'_i|=1$. Now, we will show $\bigcup_{i=0}^{k} C'_i\subseteq \bigcup_{i=0}^{k} C_i$ \\
Base case: $C'_0\subseteq C_0$.\\
Assume the contrary that $C'_0\not\subseteq C_0$. Then, $\exists g\in C'_0: g\in \bigcup_{i=1}^n C_i$. By Lemma 2 and 3, we know $g_{s_0}\in C'_0$. Furthermore, by claim 2, we have $\forall g \in C_0: g\in C'_0$. Hence, $C_0\cup {g_{s_0}}\subseteq C'_0$. However, by Lemma 1, $C_0\cup {g_{s_0}}$ is not isomorphic to any subgraph of $G$. Hence, $C'_0$ can't be executed under the same mapping. This is a contradiction.\\
Inductive step: \\
Assume $\bigcup_{j=0}^i C'_j \subseteq \bigcup_{j=0}^i C_j$. We want to show $\bigcup_{j=0}^{i+1} C'_j \subseteq \bigcup_{j=0}^{i+1} C_j$. \\
Again assume the contrary that $\bigcup_{j=0}^{i+1} C'_j \not \subseteq \bigcup_{j=0}^{i+1} C_j$. Then, $(\bigcup_{j=0}^{i} C'_j)\cup C'_{i+1} \not \subseteq \bigcup_{j=0}^{i+1} C_j$. Since $\bigcup_{j=0}^i C'_j \subseteq \bigcup_{j=0}^i C_j$, $\exists g\in C'_{i+1}: g\in \bigcup_{j=i+2}^n C_j$. By Lemma 2 and 3, we have $g_{s_{i+1}}\in C'_{i+1}$. Also by Lemma 2, we have $C_{i+1}\subseteq C'_{i+1}$. Thus, $C_{i+1}\cup \{g_{s_{i+1}}\}\subseteq C'_{i+1}$. By Lemma 1, this $C'_{i+1}$ is not isomorphic to any subgraph of $G$, which is a contradiction. \\
Since $k<n$ and we prove by induction that $\bigcup_{i=0}^{k} C'_i\subseteq \bigcup_{i=0}^{k} C_i$, we know $\bigcup_{i=0}^{k} C'_i \neq C$. This means not all gates in $C$ are executed in $C'$. Hence, there doesn't exist a transpiled circuit with $k<n$ SWAP gates. 
\end{proof}
 Since we need at least $n$ SWAP gates to execute $C$ and we know a solution with $n$ SWAP, the SWAP count is optimal. Lastly, since the final circuit is obtained by inserting more redundant gates to the backbone, the non-isomorphic and dependency properties still hold after insertion because adding more edges to the interaction graph will keep it non-isomorphic, and the dependency relation won't break by inserting more gates into the sequence since it can only make the dependency chain either longer or remain the same. We still need at least $n$ SWAP gates to execute $C$, and hence  $C_{ans}$ is optimal since it requires $n$ SWAP gates.

\section{Evaluation}
\label{sec:evaluation}

\begin{figure}[!htbp]
\centering
\subfigure[]{
\includegraphics[width=0.9\linewidth]{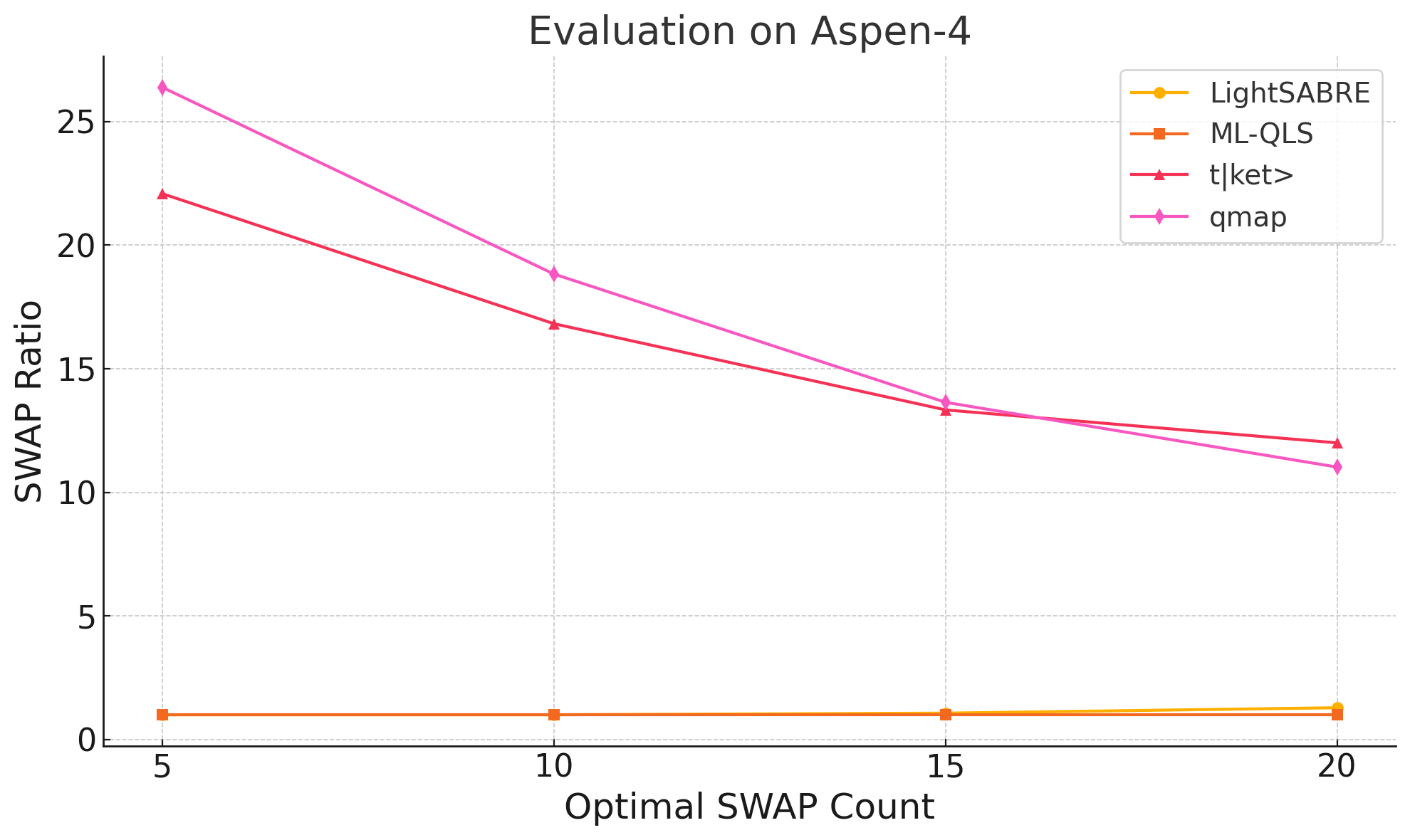}
}
\hfill
\subfigure[]{
\includegraphics[width=0.9\linewidth]
{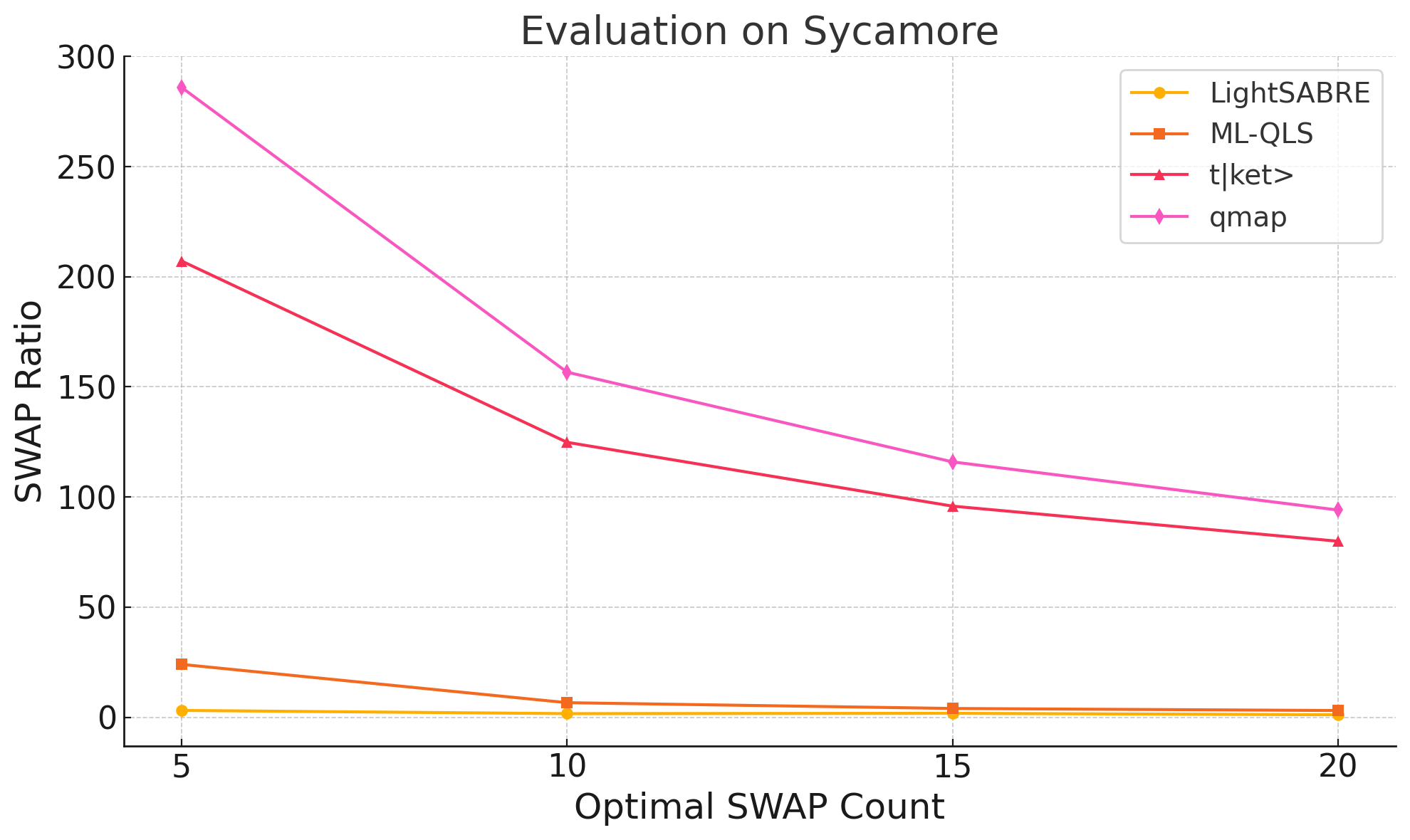}
}
\subfigure[]{
\includegraphics[width=0.9\linewidth]
{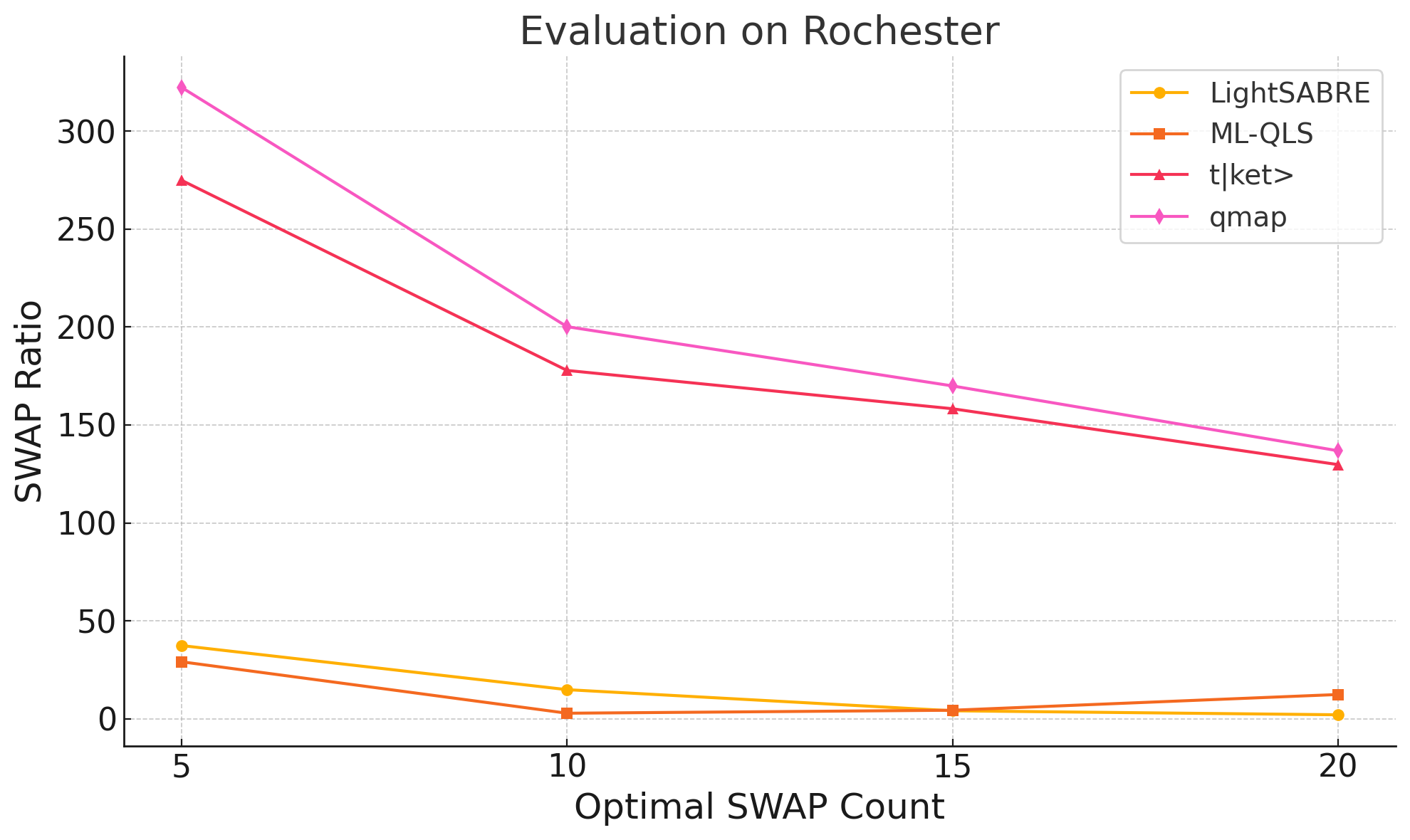}
}
\subfigure[]{
\includegraphics[width=0.9\linewidth]
{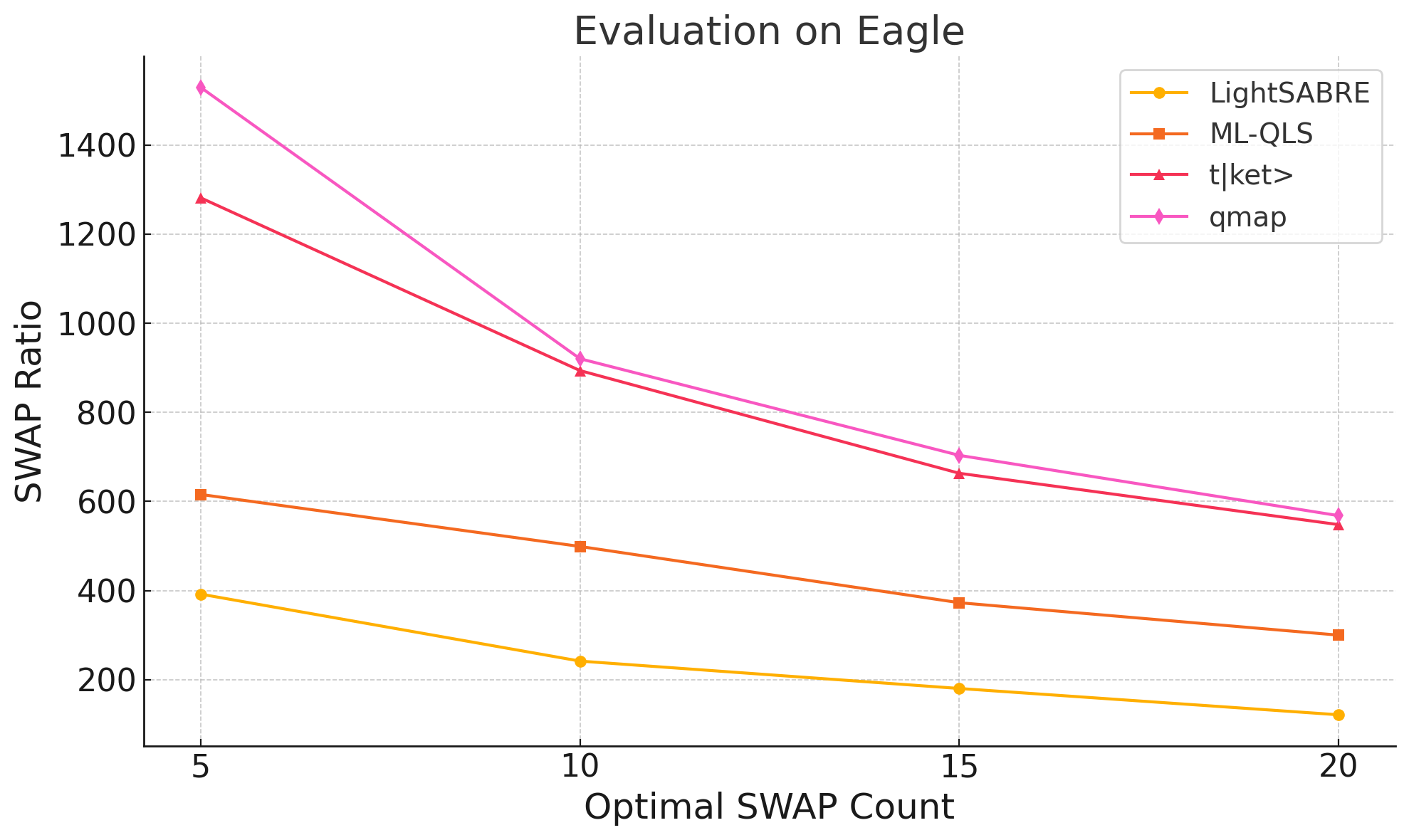}
}
\caption{Results for the heuristic experiments. (a) shows the result for the Aspen-4 architecture. (b) reveals the result for the Sycamore architecture. (c) is the result for the Rochester architecture. (d) shows the result for the Eagle architecture.}
\label{fig:exp2results}
\end{figure}

We verify the optimality of QUBIKOS in Section~\ref{subsec:optimal_study} via OLSQ2~\cite{lin2023olsq2} with PySAT version 1.8.dev13 and Z3-solver version 4.13.3.0.
In Section~\ref{subsec:qls_eval}, we evaluate the optimality gap of existing QLS tool, LightSABRE from Qiskit~\cite{zou2024lightsabre} version 1.2.4, t$\ket{\text{ket}}$~\cite{tket} version 1.34.0, QMAP~\cite{qmap} version 2.7.0, and ML-QLS~\cite{lin2024mlqls}.
All the experiments were conducted on 
Intel Xeon CPU at 2.20GHz.
For LightSABRE, we test it with 1000 trials.

\subsection{Optimality Study}
\label{subsec:optimal_study}

This experiment is to verify the optimal SWAP count for QUBIKOS circuits via the exact QLS tool OLSQ2~\cite{lin2023olsq2}.
We generate QUBIKOS circuits for the Rigetti Aspen-4 with 16 qubits and 3x3 grid architectures. 
We create 400 circuits per architecture, with 100 for each SWAP count from 1 to 4.
To guarantee a feasible runtime given OLSQ2's exponential complexity, we limit circuits to 30 two-qubit gates.

The optimality experiment confirms that all generated circuits have the desired SWAP counts. 
For each architecture, circuits designed with specific SWAP counts were verified by OLSQ2 to require the same SWAP counts, with no deviations observed. 
This consistency demonstrates that the benchmark circuits are optimally constructed.

\subsection{QLS Tool Evaluation}
\label{subsec:qls_eval}

This experiment is to measure the optimality gap of leading heuristic QLS tools.
We produce QUBIKOS circuits for four architectures: Rigetti Aspen-4, Google Sycamore with 54 qubits, IBM Rochester with 53 qubits, and IBM Eagle with 127 qubits. 
For each architecture, we make 40 circuits with 10 for each SWAP count of 5, 10, 15, and 20. 
The two-qubit gate count is set to 300 for Aspen-4, 1500 for Sycamore and Rochester, and 3000 for Eagle. The reason for varying circuit size on different architectures is because a larger architecture requires more gates on average to construct a section of the backbone circuit as the interaction graph requiring more connections on average to be non-isomorphic. 

The optimality gap is defined based on the SWAP ratio, which is calculated as the ratio of the average SWAP count to the optimal SWAP count: $\frac{\mathit{Average\; SWAP\; count}}{\mathit{Optimal\; SWAP\; count}}$. 
This ratio is always greater than or equal to 1, where a ratio of 1 indicates that the QLS tool finds the optimal solution.

Figure~\ref{fig:exp2results} presents the experimment results, showing how the optimality gap varies across different QLS tools.
Among the four QLS tools, ML-QLS achieves optimal results on Aspen-4 and Rochester, 
and LightSABRE performs the best on the other architectures.
ML-QLS performs similarly to LightSABRE, except on the Eagle architecture.
This is different from the results in the ML-QLS paper \cite{lin2024mlqls} because we uses a newer version of Qiskit and much higher trial number. So, LightSABRE has a better performance in our experiment but not in the ML-QLS experiment.
In contrast, t$\ket{\text{ket}}$ and QMAP exhibit a significantly larger optimality gap with an average of 185x and 207x respectively on Aspen-4 and as high as 846x and 930x on Eagle.

As the number of physical qubits increases, the optimality gap also tends to grow. This is supported by the fact that on Aspen-4, ML-QLS finds the optimal solutions in all cases. But on Sycamore, Rochester, and Eagle, LightSABRE has an optimality gap of 1.95x, 12.17x, and 233.97x respectively. 
However, the number of physical qubits is not the only factor; the connectivity of the coupling graph also plays a significant role.
For example, despite the Rochester architecture having a similar number of physical qubits to the Sycamore architecture, its optimality gap is six times larger. 
This is due to Rochester’s heavy-hex structure, which has sparse connectivity and has fewer axes of symmetry compared to Sycamore's dense grid connectivity. 
As a result, architectures with dense connectivity are more favorable for QLS tools, leading to smaller optimality gaps.

These findings highlight the utility of our examples for evaluating QLS tools, revealing that the optimality gap varies significantly depending on the tool. 
While LightSABRE and ML-QLS have smaller gaps, t$\ket{\text{ket}}$ and QMAP struggle to achieve near-optimal solutions. 
This emphasizes the need for more efficient QLS tools, particularly for larger architectures where the current tools demonstrate a large optimality gap. 
Additionally, the performance differences between the architectures suggest that the connectivity of the coupling graph influences how well QLS tools can minimize the SWAP count, with tools performing better on more connected architectures.

\subsection{Case Study: LightSABRE}
QUBIKOS not only provides a consistent evaluation of existing QLS tools but also offers insights for improving them. 
In this section, we perform a case study to analyze the results obtained from LightSABRE.


In one Aspen-4 experiment, SABRE identifies the optimal initial mapping but fails to find the correct routing, leading to a non-optimal solution. 
Figure~\ref{fig:circuit} illustrates the partial circuit where SABRE deviates from the optimal routing. 
Instead of choosing to SWAP($q_3, q_9$), LightSABRE opts for SWAP($q_2, q_9$), resulting in additional SWAP gates and non-isomorphic mappings. 
To understand why LightSABRE chooses this suboptimal SWAP, we calculate the cost of the two SWAP gates. 
SABRE evaluates three types of costs: basic, lookahead, and decay. 
In this case, both SWAPs are used for the first time and connect $q_2$ and $q_3$, so the basic and decay costs are the same. 
The difference lies in the lookahead cost. 
In LightSABRE, the lookahead is controlled by the extended set, where its size is the number of gates LightSABRE will consider beyond the execution layer. In Qiskit, the extended set size is 20 with a weight of 0.5. 
Thus, the lookahead cost for SWAP($q_3, q_9$) is 0.7, while for SWAP($q_2, q_9$), it is 0.65. 
This occurs because the qubit distance reduction for SWAP($q_2, q_9$) is greater when considering the first 20 gates equally, even though the two mappings cannot be executed under the same conditions, resulting in a non-optimal choice.

This issue could be mitigated by adding a decay factor to the lookahead cost, placing more weight on gates closer to the execution layer in $D$ while reducing the weight for gates farther from execution. 
With an appropriate decay factor, SABRE could find the optimal routing.

\begin{figure}[t]
\centering
\subfigure[]{\includegraphics[width=0.38\textwidth]{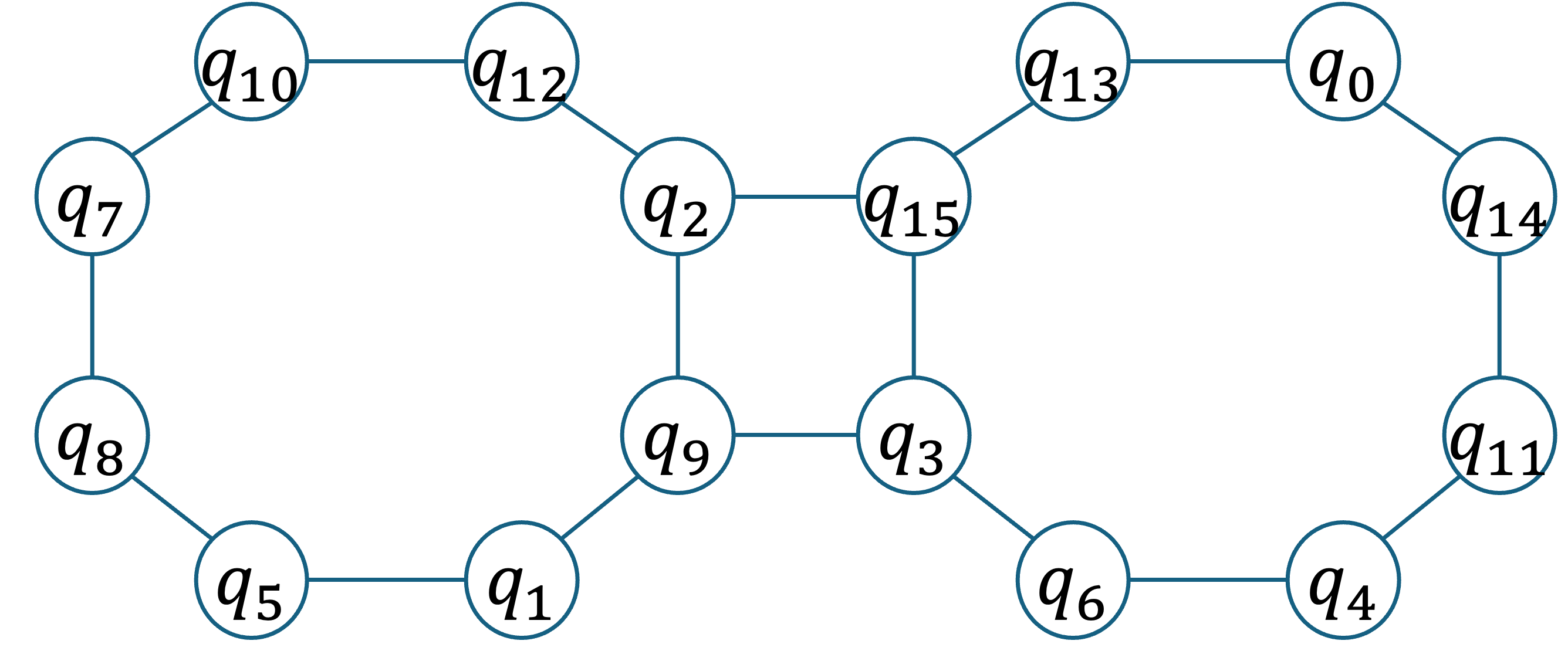}
}
\subfigure[]{
\includegraphics[width=0.4\textwidth]{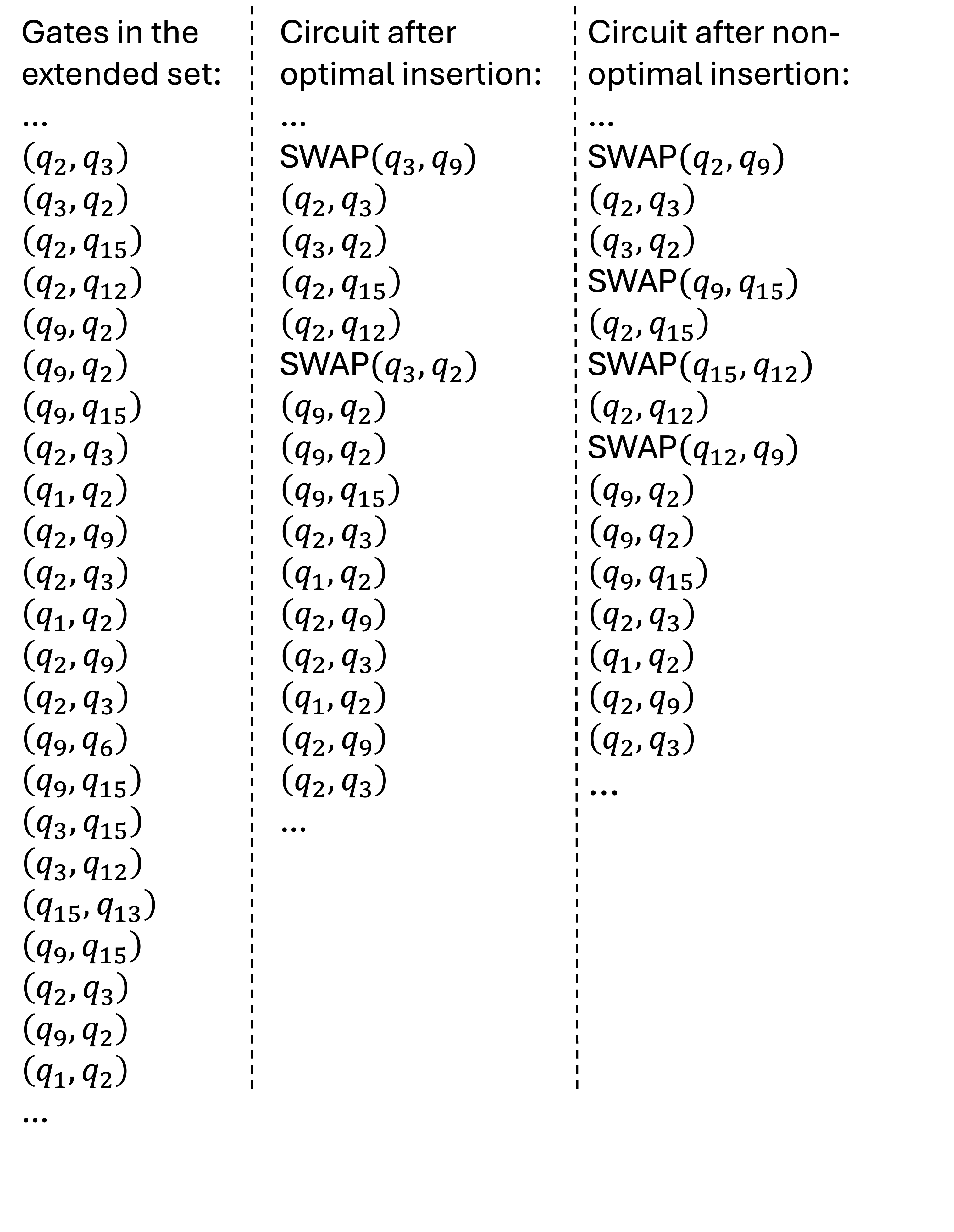}
}

\caption{\label{fig:circuit}(a).Qubit mapping before LightSABRE generates non-optimal SWAP gate. (b). Portion of the circuit that LightSABRE inserts non-optimal SWAP. The left circuit is the benchmark circuit before insertion. It includes 20 gates due to the extended set size.  The middle one is the circuit with optimal SWAP insertion. The right one is the non-optimal results from SABRE. }
\end{figure}

The above analysis is specific to SABRE, but it can be done with any QLS tools to improve its performances. Moreover, this example reveals QUBIKOS circuit has a non-trivial qubit routing even when the initial mapping is optimal. This means QUBIKOS can also be utilized to evaluate standalone routers like \cite{alpharouter} that require an initial mapping as input. Measuring the performance of these routers can be difficult since the quality of the initial mapping will directly impact its performance. With QUBIKOS, we can assess the routers by providing the initial mapping of the generated circuit. This way, we can test the routers with the optimal initial mapping, and any non-optimal results from the routers directly relates to the design of the router itself rather than the initial mapping.  

\section{Conclusion}
\label{sec:conclusion}
In this work, we introduce QUBIKOS, a benchmark with a known optimal SWAP count. We detail the method for generating QUBIKOS circuit and provide an overview for the proof. 
We verify its optimality experimentally and use it to evaluate the performance of four QLS tools. 
Our experiment results show that LightSABRE delivers the best performance but there is still an optimality gap that grows with the size of the architecture. 
Additionally, we illustrate how QUBIKOS can aid in improving QLS tools by analyzing a non-optimal case in LightSABRE, highlighting its potential to drive the development of more efficient QLS tools in the future.

\bibliographystyle{IEEEtran}
\bibliography{references}

\end{document}